\documentclass[aps,prd,10pt,superscriptaddress,preprint,showpacs]{revtex4-1}

\usepackage[utf8]{inputenc}							
\usepackage{amsmath,amsthm,amssymb}					
\usepackage{hyperref}								
\usepackage{lmodern}								
\usepackage{enumerate}								
\usepackage{enumitem}
\usepackage{graphicx}								

\newtheorem{definition}{Definition}
\newtheorem{theorem}{Theorem}

\newcommand{\bbR}{\mathbb{R}}
\newcommand{\bbZ}{\mathbb{Z}}
\newcommand{\scri}{\mathcal{I}}
\newcommand{\calM}{\mathcal{M}}
\newcommand{\tr}{\operatorname{tr}}
\newcommand{\del}{\partial}
\newcommand{\artanh}{\operatorname{artanh}}

\begin{document}


\title{Horizon complementarity in elliptic de Sitter space}

\author{Lucas Hackl}
\email{lucas.hackl@psu.edu}
\affiliation{Institute for Gravitation \& the Cosmos and Physics Department, Penn State, University Park, PA 16802, USA}

\author{Yasha Neiman}
\email{yashula@gmail.com}
\affiliation{Perimeter Institute for Theoretical Physics, 31 Caroline Street N, Waterloo, ON, N2L 2Y5, Canada}

\date{\today}

\begin{abstract}
We study a quantum field in elliptic de Sitter space $dS_4/\bbZ_2$ -- the spacetime obtained from identifying antipodal points in $dS_4$. We find that the operator algebra and Hilbert space cannot be defined for the entire space, but only for observable causal patches. This makes the system into an explicit realization of the horizon complementarity principle. In the absence of a global quantum theory, we propose a recipe for translating operators and states between observers. This translation involves information loss, in accordance with the fact that two observers see different patches of the spacetime. As a check, we recover the thermal state at the de Sitter temperature as a state that appears the same to all observers. This thermal state arises from the same functional that, in ordinary $dS_4$, describes the Bunch-Davies vacuum.
\end{abstract}


\maketitle

\section{Introduction and summary}

Causal horizons hide parts of spacetime from an observer's view. In the presence of strong curvature, one often finds that \emph{no} observer has access to the entire spacetime. This is the case for black holes, as well as for de Sitter space with its cosmological horizons. A useful prejudice in physics says that a theory should only contain objects that are in principle measurable. By this standard, if no one can see the entire world, then we shouldn't speak about its quantum state. Moreover, if the world-picture available to each observer is the maximal attainable one, then there shouldn't be any \emph{additional} information outside the observer's patch: the descriptions of the different causal patches should somehow be equivalent. In particular, the two sides of each horizon should contain the same information. Together, these ideas form the horizon complementarity principle \cite{Susskind:1993if,Dyson:2002pf}, proposed as a resolution to the black hole information paradox \cite{Hawking:1976ra}.

In normal circumstances, quantum field theory (QFT) seems to know nothing about these considerations: one can always define a Hilbert space on a Cauchy slice, regardless of whether the entire slice is observable. Horizon complementarity is then delegated to the domain of quantum gravity, which we do not understand well enough to explore the idea in depth. 

In this paper, we explore a peculiar spacetime, in which something very similar to horizon complementarity applies to ordinary QFT. The spacetime is elliptic de Sitter space $dS_4/\bbZ_2$ -- the quotient of ordinary de Sitter space $dS_4$, i.e. the hyperboloid $x^\mu x_\mu=1$ in $\bbR^{1,4}$, by the antipodal map $x\leftrightarrow -x$. It was proposed in \cite{Parikh:2002py,Parikh:2004ux,Parikh:2004wh} as a preferred alternative to ordinary $dS_4$ for understanding quantum gravity with a positive cosmological constant. Elliptic de Sitter space is globally not time-orientable, but it doesn't contain closed timelike loops. Thus, the non-orientability cannot be detected by any observer. In fact, one can make the case that $dS_4/\bbZ_2$ is no less realistic than global $dS_4$. In all the applications to real-world cosmology, one only uses a half of $dS_4$. Instead of ignoring the other half, one might as well identify it with the first one.

The antipodal map in $dS_4$ interchanges the inside and the outside of every cosmological horizon. Thus, in $dS_4/\bbZ_2$, the two sides of each horizon are the same -- globally, the horizon has only one side. This is a very literal realization of the idea that the two sides should encode the same information. More surprisingly, we will see that the global non-orientability in time precludes the standard quantization of a field on $dS_4/\bbZ_2$. Instead, we can only quantize in a region where a positive time direction can be chosen consistently. In our geometry, such regions coincide with the causal patches of observers. Thus, the system is realizing another element of horizon complementarity -- a quantum description only exists for individual causal patches. 

This puts us in a fascinating situation. On one hand, we are outside conventional quantum mechanics, making contact with some of the deep issues in quantum gravity. On the other hand, we are dealing with ordinary, even free, QFT, where the equations are under full control. The immediate question to ask is -- if the observers are all associated with different Hilbert spaces, can we relate their world-pictures? We will present and motivate a particular recipe for translating operators and states between observers, making use of a global structure that is more primitive than a Hilbert space or an operator algebra. 

The above statements apply to any QFT on $dS_4/\bbZ_2$ (although, to be well-defined on this background, the theory must be parity-invariant \cite{Neiman:2014npa}). For simplicity, we will mostly consider a conformally coupled massless scalar $F$, with field equation:
\begin{align}
 \Box F = 2F \ , \label{eq:F_diff}
\end{align}
where we used the Ricci scalar $R=12$ in de Sitter space of unit radius. For this particular theory, we will be able to go one step further and ask: what state will appear the same to all observers under our state-translation recipe? A non-trivial calculation shows that the answer is the thermal state at $T = 1/2\pi$, which is the correct temperature for the de Sitter cosmological horizon \cite{Gibbons:1977mu}. We view this as a successful check of our proposed framework.

The rest of the paper is structured as follows. In section \ref{sec:geometry}, we review in further detail the geometry of $dS_4/\bbZ_2$ and the causal patches associated with observers. In section \ref{sec:field}, we review the solutions of the field equation \eqref{eq:F_diff}, which we parametrize in terms of asymptotic boundary data. In section \ref{sec:quantum}, we show that this solution space is not a phase space, since it doesn't admit Poisson brackets that preserve the de Sitter symmetries. One can only define Poisson brackets in a causal patch, after breaking the symmetry and choosing an observer. Upon quantization, this leads directly to a separate Hilbert space for each observer. In section \ref{sec:operators}, we propose a recipe for translating operators between these Hilbert spaces. In section \ref{sec:HH}, we apply this recipe to produce an operator that looks the same to all observers, based on the Euclidean action functional. In section \ref{sec:states}, we use the result to motivate a translation recipe for states. The above operator then becomes the de Sitter thermal state. In section \ref{sec:discuss}, we conclude and discuss future directions.

\section{Elliptic de Sitter space} \label{sec:geometry}

\subsection{Topology, observers and horizons}

We define de Sitter space $dS_4$ as the hyperboloid of unit spacelike radius in 4+1d flat space $\bbR^{1,4}$ with metric signature $(-,+,+,+,+)$:
\begin{align}
  dS_4 = \{x^\mu\in\bbR^{1,4}|x^\mu x_\mu=1\} \ .
\end{align}
The isometry group $SO(4,1)$ of $dS_4$ is then realized as the rotation group in the 4+1d space. The asymptotic boundary of $dS_4$ is composed of two conformal 3-spheres -- one at past infinity ($\scri^-$) and the other at future infinity ($\scri^+$). In the 4+1d picture, $\scri^\pm$ correspond to the 3-spheres of future-pointing and past-pointing null directions. Local Weyl transformations on $\scri^\pm$ correspond to rescalings $\ell^\mu\to \lambda\ell^\mu$ of the lightcone in $\bbR^{1,4}$; for a field on $\scri^\pm$ with conformal weight $\Delta$, this induces the rescaling $f\to \lambda^{-\Delta}f$.

Elliptic de Sitter space $dS_4/\bbZ_2$ is obtained by identifying antipodal points $x^\mu\leftrightarrow-x^\mu$. See figures \ref{fig:edS-3d}-\ref{fig:edS-2d}. The resulting spacetime still looks locally like $dS_4$, but its global properties are different. In particular, $\scri^-$ becomes identified with $\scri^+$, so the asymptotic boundary is now a single 3-sphere $\scri$. Topologically, de Sitter space $dS_4$ is a cylinder $\bbR\times S^3$, where the $S^3$ boundaries correspond to $\scri^\pm$. One way to obtain elliptic de Sitter space $dS_4/\bbZ_2$ is to cut the cylinder in half along an equatorial $S^3$ slice, and then fold this $S^3$ into a $S^3/\bbZ_2$ by identifying antipodal points. See figure \ref{fig:edS_intrinsic}. Thus, we can single out two preferred types of spatial slices in $dS_4/\bbZ_2$. One is the asymptotic slice $\scri$, with topology $S^3$. The other is an equatorial slice, with topology $S^3/\bbZ_2$ (there is a 4d set of such slices, corresponding to spacelike 4-planes through the origin in $\bbR^{1,4}$).
\begin{figure}[thp]
\centering
\includegraphics{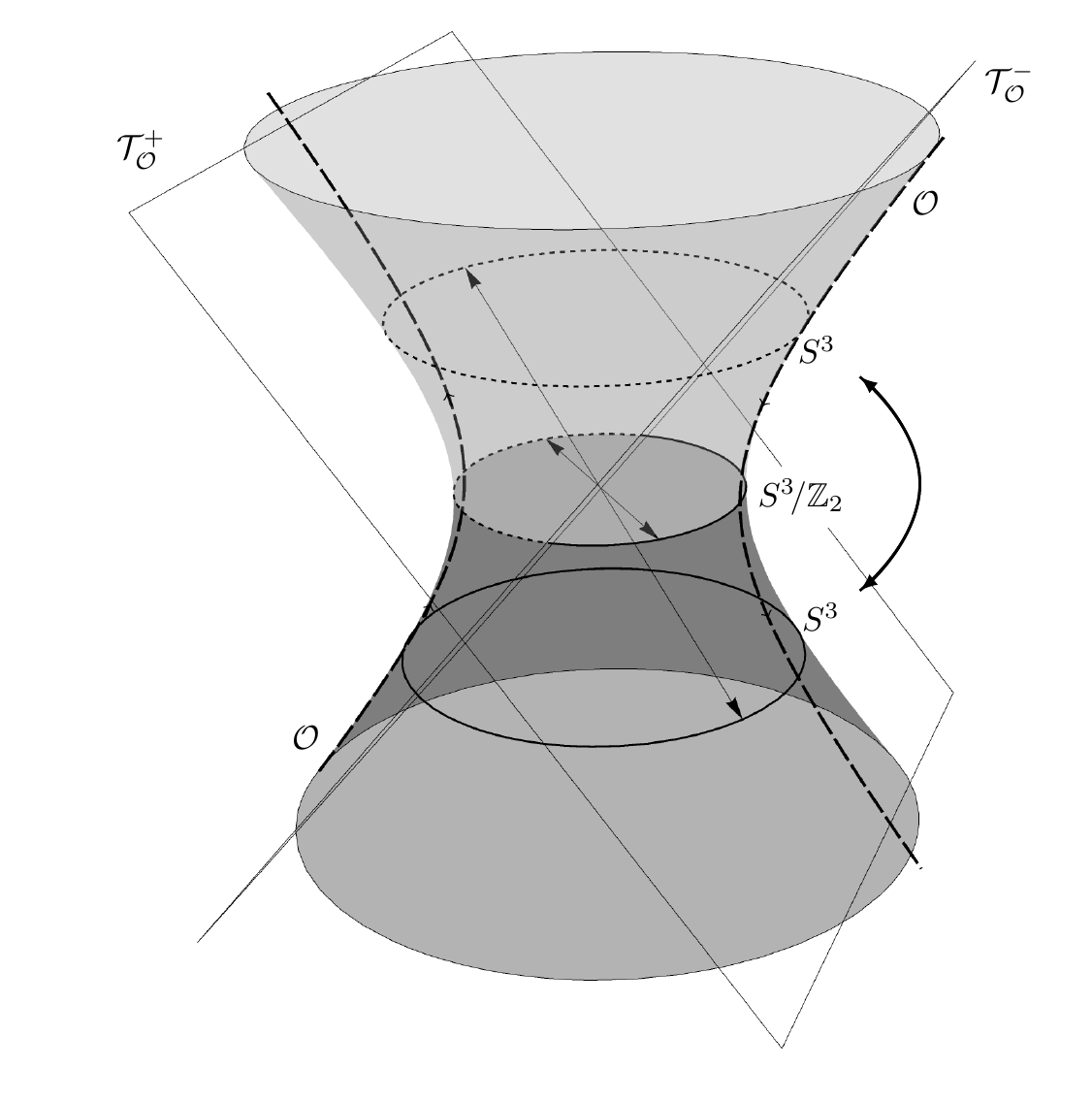}
\caption{This illustrates $dS_4$ embedded in $\mathbb{R}^{4,1}$. To get $dS_4/\mathbb{Z}_2$, the upper half (light gray) and the lower part (dark gray) are identified via the antipodal map, illustrated by black arrows. We also show an observer $\mathcal{O}$ together with his future and past horizons $\mathcal{T}^\pm_{\mathcal{O}}$. The rings in the figure represent $3$-spheres; the central ring becomes an $S^3/\mathbb{Z}_2$ due to the antipodal identification.}
\label{fig:edS-3d}
\end{figure}
\begin{figure}[thp]
\centering
\includegraphics{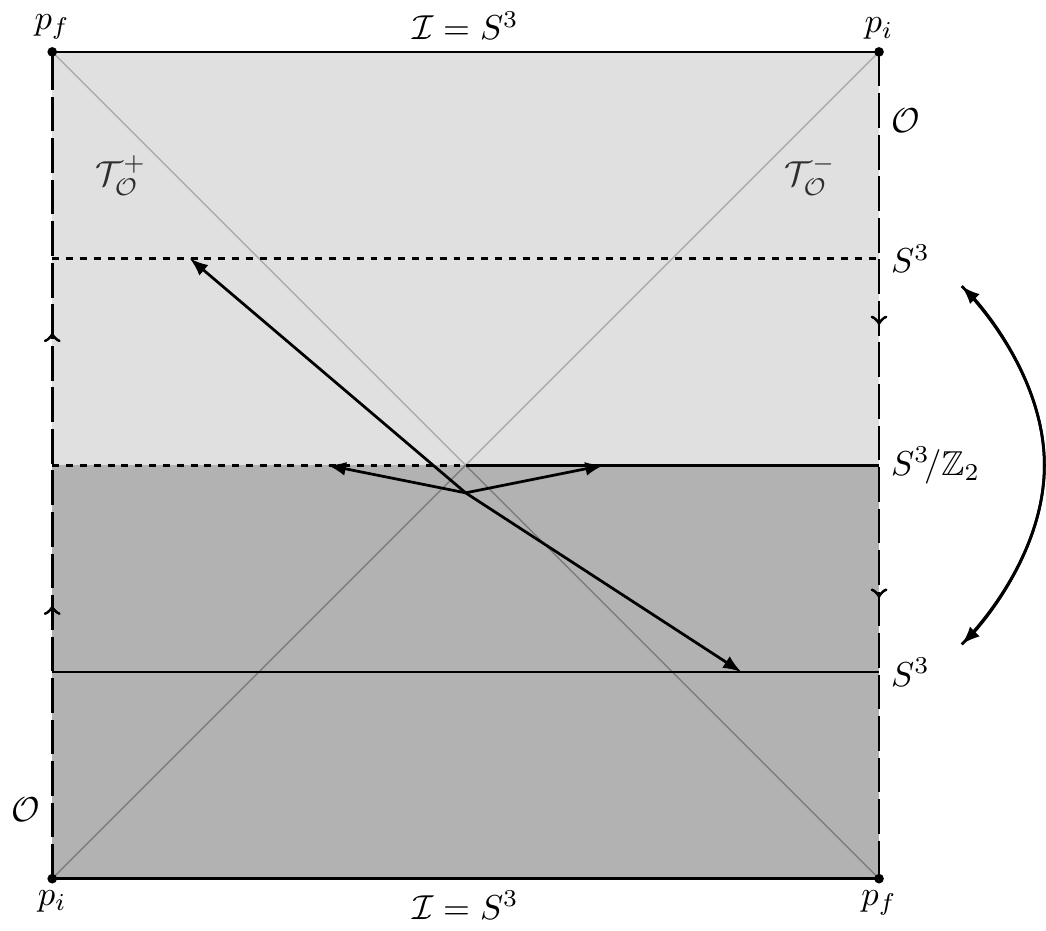}
\caption{This is the Penrose diagram corresponding to figure \ref{fig:edS-3d}. It shows the full $dS_4$ spacetime, along with the antipodal identification that turns it into $dS_4/\mathbb{Z}_2$, again illustrated by arrows. A point in this diagram represents a $2$-sphere, except on the left and right boundary (worldline of the observer $\mathcal{O}$), where a point represents an actual point. The Penrose diagram also shows the identified future and past infinity $\mathcal{I}$.}
\label{fig:edS-2d}
\end{figure}
\begin{figure}[thp]
\centering
\includegraphics{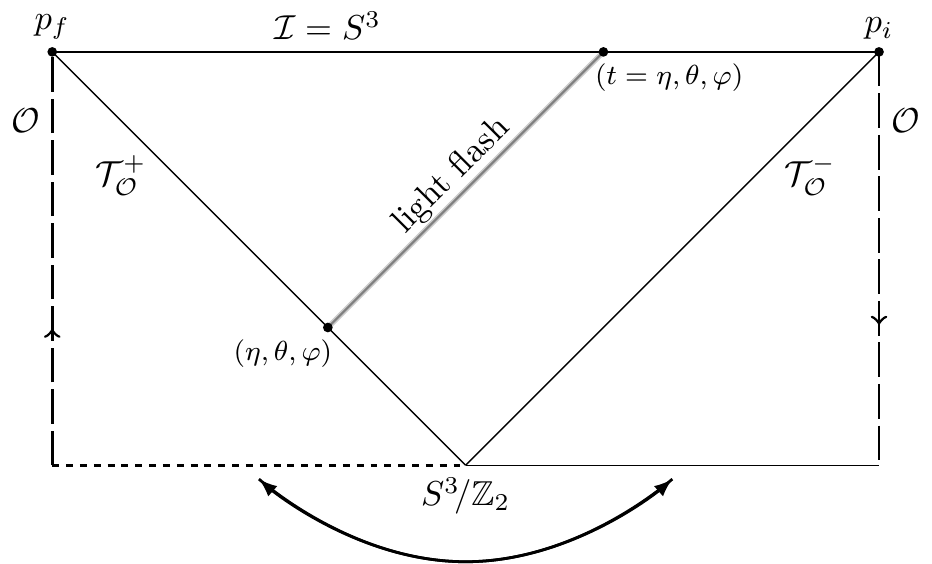}
\caption{This Penrose diagram shows all of $dS_4/\mathbb{Z}_2$, without embedding into $dS_4$. It has the topology of a cylinder where one of the two $S^3$ boundaries was folded into an $S^3/\mathbb{Z}_2$. Thus, elliptic de Sitter space has only a single boundary $\mathcal{I}\cong S^3$. The diagram also shows a light flash that starts on the observer's future horizon at $(\eta,\theta,\phi)$ and hits $\scri$ at $(t=\eta,\theta,\varphi)$, thus relating the coordinate $\eta$ on the horizon with $t$ on $\scri$.}
\label{fig:edS_intrinsic}
\end{figure}

An observer in $dS_4$ follows a timelike worldline, which begins at some point $p_i\in\scri^-$ and ends at some point $p_f\in\scri^+$. These two points can be encoded as two null directions in $\bbR^{1,4}$. The lightcones of $p_{i,f}$ are the observer's past and future cosmological horizons, respectively. They refocus at the antipodal points $-p_i$ and $-p_f$ on $\scri^\pm$. The horizons intersect at a 2-sphere, known as the bifurcation surface. They divide the spacetime into four quadrants; the quadrant containing the observer's worldline is known as his static patch. The observer in $dS_4$ can see half of the spacetime, up to his future horizon, and can influence half of the spacetime, up to his past horizon. Though we will picture the observer as traveling along the \emph{geodesic} from $p_i$ to $p_f$, this is not essential: any worldline with the same two endpoints is associated with the same horizons. 

The points $p_{i,f}$ also single out a particular Killing field $\xi^\mu$ from the de Sitter symmetry group $SO(4,1)$. The vector field $\xi^\mu$ is future-pointing timelike in the static patch, past-pointing timelike in its antipode, spacelike in the past and future quadrants, null on the horizons, and vanishing on the bifurcation surface. See figure \ref{fig:killing}. The geodesic from $p_i$ to $p_f$ runs along $\xi^\mu$. In the 4+1d picture, $\xi^\mu$ corresponds to the boost generator in the plane of the two null directions $p_{i,f}$. The full residual symmetry after picking the $p_{i,f}$ is $\bbR\times O(3)$, where the $\bbR$ corresponds to the ``time translations'' generated by $\xi^\mu$, and the $O(3)$ is the group of spatial rotations and reflections around the $p_i\rightarrow p_f$ geodesic.
\begin{figure}[thp]
\centering
\includegraphics{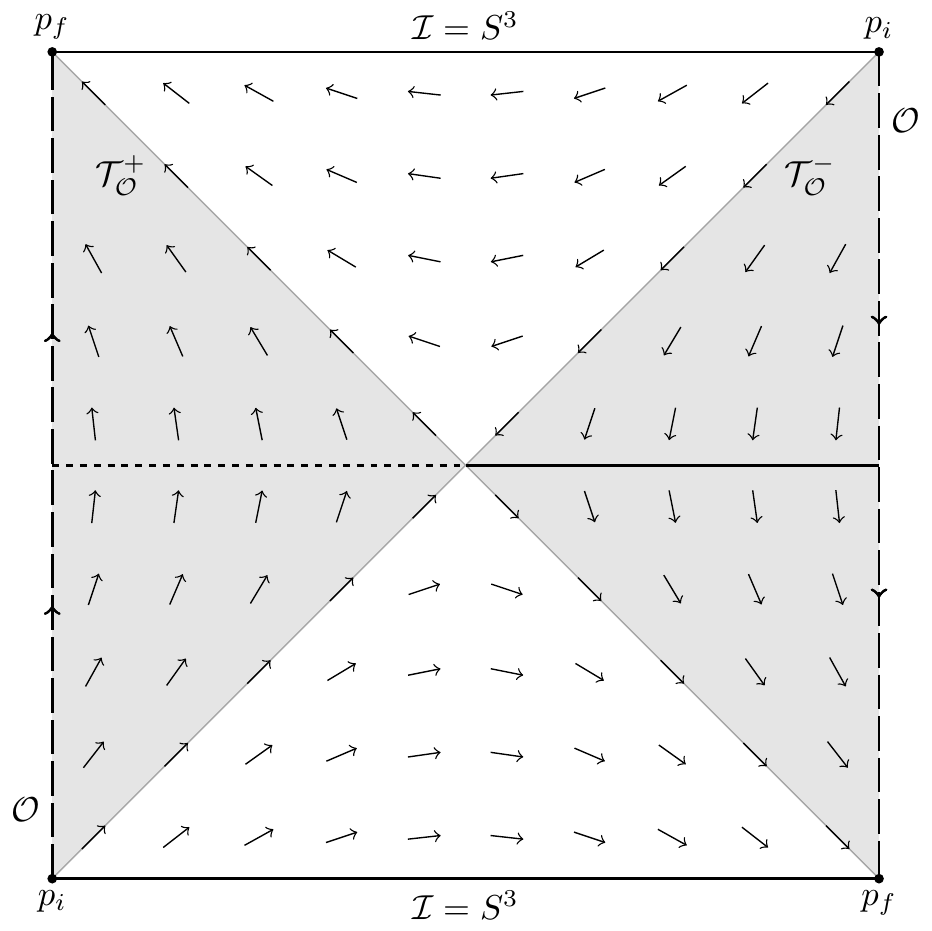}
\caption{This Penrose diagram shows the Killing vector field $\xi$ that is timelike in the observer's causal diamond (dark region) and spacelike in its complement. It becomes null on the observer's future and past horizons $\mathcal{T}^\pm_\mathcal{O}$. On the horizon, we have $\xi=\partial_\eta$ in the coordinates $(\eta,\theta,\varphi)$; on $\scri$, we have $\xi=\partial_t$ in the coordinates $(t,\theta,\varphi)$.}
\label{fig:killing}
\end{figure}

In $dS_4/\bbZ_2$, the two endpoints $p_{i,f}$ of the observer's worldline are both on the same asymptotic 3-sphere $\scri$. We can still specify one point as initial and the other as final, thus choosing a direction for the observer's proper time. The lightcones of $p_{i,f}$ still function as horizons, each one now refocusing back at its starting point. Each horizon is antipodally identified with itself, so that the only independent horizon points are from $p_i/p_f$ up to the bifurcation surface. The bifurcation surface is also identified with itself, becoming an $S^2/\bbZ_2$. The two sides of each horizon are identified with each other. As a result, the observer can now see the entire spacetime except his own future horizon, and can influence the entire spacetime except his past horizon. We see that this geometry fits well with the ideas of horizon complementarity. On one hand, the two sides of each horizon carry the same information. On the other hand, different observers have access to \emph{almost} the same data: each one sees the same spacetime, but with a different surface subtracted.

\subsection{Observer-adapted coordinates} \label{sec:geometry:coordinates}

Consider an observer in $dS_4/\bbZ_2$. We can choose a frame in $\bbR^{1,4}$ so that the endpoints $p_{i,f}$ of the observer's worldline correspond to the null directions $(1,\mp 1,0,0,0)$, where we've used the antipodal map to make both null directions future-pointing. We can now choose coordinates in $dS_4/\bbZ_2$ adapted to our observer, e.g. the static coordinates in his static patch. In the following, we will only need explicit coordinates on $\scri$ and on the observer's future horizon. 

We can represent $\scri$ as the following section of the lightcone in $\bbR^{1,4}$:
\begin{align}
 \ell^\mu(t,\theta,\varphi)=(\cosh t,\sinh t,\cos\theta,\sin\theta\cos\varphi,\sin\theta\sin\varphi) \ , \label{eq:ell}
\end{align}
where $(t,\theta,\varphi)$ are coordinates on $\scri$, adapted to the residual $\bbR\times O(3)$ symmetry. $t$ ranges from $-\infty$ at $p_i$ to $\infty$ at $p_f$. The coordinates \eqref{eq:ell} are singular at $p_{i,f}$, making the topology of $\scri$ appear like a cylinder $\bbR\times S^2$. In fact, the metric of the section \eqref{eq:ell} is also that of a cylinder:
\begin{align}
 d\ell^\mu d\ell_\mu &= dt^2 + d\theta^2 + \sin^2\theta\, d\varphi^2 \quad ; \quad d^3\ell = dt\,d(\cos\theta)\,d\varphi \ . \label{eq:d_ell}
\end{align}
When performing calculations for the given observer, this will be our chosen representative of the conformal class of metrics on $\scri$. Translations in the ``time'' coordinate $t$ are generated by the Killing field $\xi^\mu$, which generates time translations in the static patch. On $\scri$, the vector $\xi^\mu$ and the coordinate $t$ are spacelike.

Similarly, we will use coordinates $(\eta,\theta,\varphi)$ on the observer's future horizon, such that a point on the horizon is represented in $\bbR^{1,4}$ as:
\begin{align}
 x^\mu = (e^\eta,e^\eta,\cos\theta,\sin\theta\cos\varphi,\sin\theta\sin\varphi) \ . \label{eq:horizon}
\end{align}
The coordinate $\eta$ ranges from $-\infty$ at the bifurcation surface to $\infty$ at $p_f$. Again, the Killing field $\xi^\mu$ generates translations in $\eta$, while the $O(3)$ symmetry acts on $(\theta,\varphi)$. On the horizon, $\xi^\mu$ and the coordinate $\eta$ are null. The coordinates \eqref{eq:horizon} are singular at $p_f$ and at the bifurcation surface. Our coordinates for the horizon and $\scri$ are directly related: a radial light flash from the horizon point $(\eta,\theta,\varphi)$ will hit $\scri$ at $(t,\theta,\varphi)$, where $t=\eta$.

\section{Conformally coupled massless scalar} \label{sec:field}

We now introduce a real scalar field $F$, subject to the field equation \eqref{eq:F_diff}. There are two ways to define such a field on $dS_4/\bbZ_2$. One is to start with an antipodally even field in $dS_4$, which will then live in the trivial bundle on $dS_4/\bbZ_2$. The other is to start with an antipodally odd field, which will live on a twisted bundle on $dS_4/\bbZ_2$, changing its sign along incontractible cycles. We will focus on the antipodally even case. 

To discuss asymptotic boundary data, we follow the usual procedure of conformal completion: we introduce a time coordinate $z$ such that $z\to 0$ on $\scri$, and the rescaled metric $z^2 g_{\mu\nu}$ is regular at $z=0$. In a neighborhood of $\scri$, we can use $z$ together with $(t,\theta,\varphi)$ from section \ref{sec:geometry:coordinates} as coordinates on $dS_4/\bbZ_2$. The general solution of the field equation \eqref{eq:F_diff} in ordinary $dS_4$ has the asymptotic behavior:
\begin{align}
 F = z\,\Phi(t,\theta,\varphi) + z^2\,\Pi(t,\theta,\varphi) + \dots \ ,
\end{align}
where the higher-order terms in $z$ are all determined by the Dirichlet boundary data $\Phi$ and the Neumann boundary data $\Pi$. We view $\Phi$ and $\Pi$ as fields on $\scri$, with respective conformal weights 1 and 2. For an antipodally even/odd field $F$, the $\Phi$/$\Pi$ boundary data vanishes \cite{Ng:2012xp}. Thus, in our $dS_4/\bbZ_2$ setup, a solution $F(x)$ is uniquely determined by the Neumann data $\Pi$:
\begin{align}
 F(x)=\int_{\scri}\Pi(\ell)\,G(x;\ell)\,d^3\ell \ . \label{eq:F_Pi}
\end{align}
Here, $\ell^\mu(t,\theta,\varphi)$ stands for a boundary point in the parametrization \eqref{eq:ell}, and the boundary-to-bulk propagator $G(x;\ell)$ reads \cite{Neiman:2014npa}:
\begin{align}
 G(x;\ell)=\frac{1}{4\pi}\,\delta(x\cdot\ell) \ , \label{eq:propagator}
\end{align}
where $x\cdot\ell \equiv x_\mu\ell^\mu$ is the scalar product in $\bbR^{1,4}$.

We can thus identify the space of solutions $\Gamma$ with the space of Neumann boundary data $\Pi(\ell)$ on the 3-sphere $\scri$. In observer-adapted coordinates, we can parametrize this space of boundary data using Fourier modes in $t$ and spherical harmonics in $(\theta,\varphi)$:
\begin{align}
 \Pi(t,\theta,\varphi) = \int^\infty_0\frac{d\omega}{2\pi}\sum^\infty_{l=0}\sum^l_{\,m=-l}\!\left(e^{-i\omega t}\,Y_{lm}(\theta,\varphi)\,c_{lm}(\omega) + c.c.\right) \, .
 \label{eq:Pi}
\end{align}
The boundary data is now parametrized by the complex coefficients $c_{lm}(\omega)$. Due to the $\bbR\times O(3)$ symmetry, each $(\omega,l,m)$ component of \eqref{eq:Pi} yields a bulk solution with frequency $\omega$ and angular momentum numbers $(l,m)$ in the static patch. Strictly speaking, \eqref{eq:Pi} is not a function on $\scri\cong S^3$, but on the $\bbR\times S^2$ obtained by removing the two points $p_{i,f}$. However, this distinction will not bother us. Any function on $S^3$ is also a function on $\bbR\times S^2$; conversely, any function on $\bbR\times S^2$ is a singular, generalized function on $S^3$, which can be viewed as the limit of a sequence of honest functions.

For more general field theories, one will similarly find that only one of the (Dirichlet/Neumann) boundary data on $\scri$ is independent. However, the other type of boundary data will not in general vanish.

\section{Phase space and quantization} \label{sec:quantum}

\subsection{There is no dS-invariant symplectic form}

Usually, the space of boundary data (or the space of solutions) comes naturally equipped with a symplectic form, which defines the Poisson brackets and makes it into a phase space. If we start with the phase space $\Gamma_{dS_4}$ of ordinary $dS_4$ and restrict to the submanifold $\Gamma$ of antipodally symmetric fields, the restriction of the standard symplectic form vanishes. Put differently, our solution space $\Gamma$ in $dS_4/\bbZ_2$ is a Lagrangian submanifold of the ordinary phase space $\Gamma_{dS_4}$.

For the particular field theory \eqref{eq:F_diff}, this result is immediate. Indeed, the asymptotic boundary data $\Phi$ and $\Pi$ are canonical conjugates in ordinary $dS_4$, while in $dS_4/\bbZ_2$, one of them is constrained to vanish. For more general field theories, the result is easier to see in terms of an equatorial $S^3$ slice in the bulk, which becomes an $S^3/\bbZ_2$ under the antipodal identification. There, the standard symplectic form reads:
\begin{align}
 \Omega(\delta_1,\delta_2) = \int dV \left(\delta_1 F(x)\delta_2\dot F(x) - \delta_2 F(x)\delta_1\dot F(x) \right) \ . \label{eq:Omega}
\end{align}
This form vanishes on the subspace of antipodally even/odd solutions, since the $\dot F$/$F$ contributions (respectively) cancel at antipodal points. Alternatively, if one thinks in terms of $dS_4/\bbZ_2$, one would try to integrate \eqref{eq:Omega} over an $S^3/\bbZ_2$ slice. However, the expression is then ill-defined: due to the global non-orientability, one cannot consistently choose a positive time direction along which to define $\dot F$.

Not only does the standard symplectic form vanish, but it's also impossible to find a different one that would respect the de Sitter symmetries:
\begin{theorem}
 Consider the space of boundary data $f(\ell)$ with conformal weight $\Delta$ on $\scri$ (e.g. $\Pi(\ell)$ with $\Delta=2$). There is no symplectic form on this space that is invariant under the $SO(4,1)$ de Sitter group.
\end{theorem}
\begin{proof}
Let us use the representation of $\scri$ via light-like vectors $\ell^\mu\in\bbR^{1,4}$. Now, any symplectic form can be written as:
\begin{align}
 \Omega(f_1,f_2) = \int d^3\ell_1 d^3\ell_2\, f_1(\ell_1)\, f_2(\ell_2) \,G(\ell_1,\ell_2) \ ,
\end{align}
where the kernel $G(\ell_1,\ell_2)$ has conformal weight $3-\Delta$ in each argument, and is antisymmetric under $\ell_1\leftrightarrow\ell_2$. For $\Omega$ to be invariant under $SO(4,1)$, the same must be true of $G$. However, the only independent invariant that can be constructed from $\ell_1,\ell_2$ is the inner product $\ell_1\cdot\ell_2$, which is symmetric rather than antisymmetric under $\ell_1\leftrightarrow\ell_2$.
\end{proof}

\subsection{Observer-dependent phase space and quantization}

In the absence of a global symplectic form, we can still use the standard symplectic form in an observer's static patch. Consider an equatorial $S^3/\bbZ_2$ slice containing the observer's $S^2/\bbZ_2$ bifurcation surface, but with the bifurcation surface excised. The remaining region, with the topology of an open ball, is precisely the region of space that is visible to the observer. In this region, we can consistently choose a unit timelike normal (pointing in the direction of $\xi^\mu$), which becomes discontinuous at the bifurcation surface. We can then use this normal to define $\dot F$, and plug that into the symplectic form \eqref{eq:Omega}. The result agrees with the standard symplectic form in $dS_4$, restricted to the static patch.

The symplectic form thus constructed is genuinely observer-dependent. This is easiest to see for two observers whose bifurcation surfaces share an $S^3/\bbZ_2$ equatorial slice (we will say in this case that the observers ``share the $S^3/\bbZ_2$''). The two symplectic forms are then given by the integral \eqref{eq:Omega} over the same $S^3/\bbZ_2$ slice, but with opposite signs for $\dot F$ in a certain region. On the other hand, this picture shows that formally, the two observers agree on phase space volumes: in the $(F,\dot F)$ basis on the shared $S^3/\bbZ_2$, their symplectic forms are block-diagonal, and only differ in signs. While not every two observers share an $S^3/\bbZ_2$, we can always find a third observer that shares an $S^3/\bbZ_2$ with each one. Thus, the notion of phase space volume is common to all observers. 

In quantum theory, the Poisson brackets, i.e. the inverse of the symplectic form, translate directly into an operator algebra, of which the Hilbert space is a representation. Thus, our statements above concerning symplectic forms translate directly into the quantum context. There is no global operator algebra or Hilbert space, since there is no observer-independent symplectic form. Instead, each observer assigns a different algebra to the field operators $\hat F(x)$, leading to a different Hilbert space for each observer. Individually, each of these Hilbert spaces is just the ordinary Hilbert space in a static patch of $dS_4$. A state in each Hilbert space can be represented as a functional $\psi[F(x)]$ over the field values on an $S^3/\bbZ_2$. However, even for observers that share an $S^3/\bbZ_2$, one cannot compare these states directly, i.e. there is no natural map between the Hilbert spaces, since the operator algebras are different. For instance, the physical meaning of the derivative $\delta/\delta F$ acting on $\psi[F]$ is observer-dependent, since different observers assign different signs to the momentum $\dot F$.

\subsection{Decomposition into harmonic oscillators}

Consider an observer, with his associated symmetry group $\bbR\times O(3)$ of time translations and rotations. The free field \eqref{eq:F_diff} can be decomposed into modes with frequency $\omega$ and angular momentum numbers $(l,m)$. Each mode behaves as an independent harmonic oscillator. Since our observer-dependent symplectic form is just the standard symplectic form in the static patch, it induces the standard commutation relations on the associated raising and lowering operators:
\begin{align}
 \begin{split}
   &[\hat a_{lm}(\omega), \hat a^\dagger_{l'm'}(\omega')] = 2\pi\delta(\omega - \omega')\delta_{ll'}\delta_{mm'} \ ; \\
   &[\hat a_{lm}(\omega), \hat a_{l'm'}(\omega')] = [\hat a^\dagger_{lm}(\omega), \hat a^\dagger_{l'm'}(\omega')] = 0 \ .
 \end{split} \label{eq:a_commutators}
\end{align}

Now, recall the decomposition \eqref{eq:Pi} of the boundary data on $\scri$ in terms of coefficients $c_{lm}(\omega)$. From symmetry, these coefficients must coincide with the  classical counterparts $a_{lm}(\omega)$ of the lowering operators $\hat a_{lm}(\omega)$, up to normalization:
\begin{align}
 a_{lm}(\omega) = N_l(\omega)c_{lm}(\omega) \ . \label{eq:a_c}
\end{align}
The complex phases of the normalization coefficients $N_l(\omega)$ are arbitrary. Their absolute values can be computed by evolving the boundary data \eqref{eq:Pi} into the bulk, and requiring that the energy in the static patch (with respect to the Killing field $\xi^\mu$) takes the form:
\begin{align}
 E = \int^\infty_0\frac{d\omega}{2\pi}\,\sum_{lm}\omega\left|a_{lm}(\omega)\right|^2 \ . \label{eq:E}
\end{align}
The energy is easiest to evaluate on the horizon, using the coordinates \eqref{eq:horizon}. The calculation is given in Appendix \ref{app:N}. The result reads:
\begin{align}
 \left|N_l(\omega)\right|^2 =
  \left\{\begin{array}{lc}
    \displaystyle \frac{1}{2\omega} \prod_{k=1}^{l/2}\frac{(2 k-1)^2+\omega ^2}{(2 k)^2+\omega ^2} \quad & l\text{ even}\\
    \displaystyle \frac{\omega}{2(1+\omega^2)} \prod_{k=1}^{(l-1)/2}\frac{(2 k)^2+\omega ^2}{(2 k+1)^2+\omega ^2} \quad & l\text{ odd}
  \end{array}\right.\,. \label{eq:N}
\end{align}

An observer's $a$'s and $a^*$'s are just linear functionals of the field. Therefore, the $a$'s and $a^*$'s of two different observers are linear combinations of each other. This may look at first like a Bogoliubov transformation, but it isn't: the commutation relations \eqref{eq:a_commutators} are not preserved. This is just another way of saying that the two observers have different symplectic forms, and thus different operator algebras.

\section{Translating operators between observers: the Wigner-Weyl transform} \label{sec:operators}

Now that we have a separate Hilbert space for each observer, how can we translate information between them? In this section, we will take a first step by proposing a translation recipe for operators between observers. In section \ref{sec:states}, we will augment this with a translation recipe for states. 

The different observers share the same space $\Gamma$ of boundary data, or, equivalently, of classical solutions (again, if we allow singular data, then this is true despite the removed points $p_{i,f}$). The observers can therefore agree on the notion of functionals $A[\Pi(\ell)]$ over $\Gamma$. Since they also agree on a notion of phase space volume, we can perform functional integrals over $\Gamma$. Thus, in particular, the observers agree on classical probability distributions on $\Gamma$.

Given an observer with his symplectic form, $\Gamma$ becomes a proper phase space. A functional $A[\Pi(\ell)]$ is then a functional over phase space, which can be translated into an operator on the observer's Hilbert space. To do this, we use the harmonic oscillator decomposition \eqref{eq:Pi} to write our functional as $A[c_{lm}(\omega),c^*_{lm}(\omega)]$ (the normalization \eqref{eq:a_c} of the oscillators is not essential here). We then expand $A[c_{lm}(\omega),c^*_{lm}(\omega)]$ in a Taylor series. It can now be reinterpreted as an operator, once we choose an operator ordering convention for each term. We pick symmetric ordering, since this is the convention on which all observers can agree. Indeed, on one hand, the $c$'s and $c^*$'s for two observers are just linear combinations of each other, and so symmetric ordering for one set is the same as symmetric ordering for another. On the other hand, the positive-frequency condition that distinguishes $c$'s from $c^*$'s is different between observers, so they cannot agree on the meaning of e.g. normal ordering.

The procedure described above for translating functionals $A[\Pi(\ell)]$ into operators $\hat A$ is the well-known Wigner-Weyl transform for quantum mechanics in phase space; see e.g. \cite{Wigner}. The transform also works in the opposite direction: starting with an operator $\hat A$, we expand it in symmetrized products of $\hat c$'s and $\hat c^\dagger$'s, and then reinterpret the result as a functional in the $c$'s and $c^*$'s.

Thus, our recipe for translating operators between observers is as follows. Starting from an operator $\hat A_1$ on one observer's Hilbert space, we perform the Wigner-Weyl transform to obtain a functional $A[\Pi(\ell)]$, which can then be transformed again into an operator $\hat A_2$ on the other observer's Hilbert space. See figure \ref{fig:operators}. This recipe preserves Hermiticity, since Hermitian operators simply map to real functionals. It preserves the trace $\tr\hat A$, since the latter can be computed by integrating the functional $A$ over $\Gamma$. Finally, it preserves the traced product $\tr(\hat A\hat B)$, since this can again be computed by integrating the product of the two functionals $AB$. More detailed properties of the operators, such as their spectrum or their positive-definiteness, are not preserved. The operator algebra is also not preserved: e.g. for operators linear in $\hat c$'s and $\hat c^\dagger$'s, the observers' different symplectic forms translate directly into different algebras. Finally, the traced product $\tr(\hat A\hat B\hat C\dots)$ of more than two operators is not preserved: one can still write this as an integral over $\Gamma$, but the product of functionals must now be written as a Moyal $\star$-product, which depends on the symplectic form. 
\begin{figure}[thp]
\centering
\includegraphics{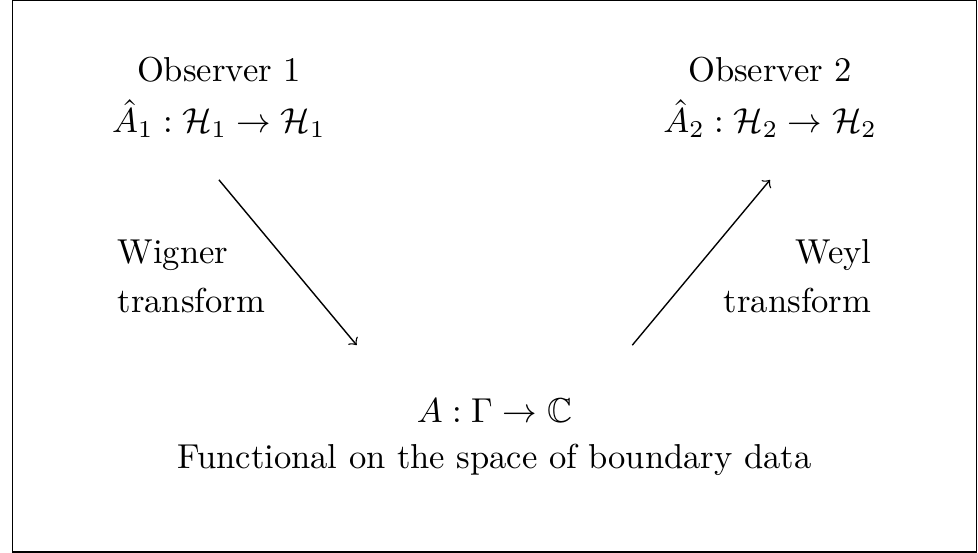}
\caption{Translation of operators between observers using the Wigner-Weyl transform.}
\label{fig:operators}
\end{figure}

Under our recipe, the field operator $\hat F(x)$ gets mapped to itself, being a linear combination of $\hat c$'s and $\hat c^\dagger$'s. The corresponding functional on the space of boundary data is the propagator \eqref{eq:F_Pi}. If we use observer-adapted coordinates, then of course $\hat F(x)$ is mapped to itself only after the appropriate coordinate transformation. A different behavior is exhibited by functionals $A[\Pi(\ell)]$ that are invariant under the de Sitter group $SO(4,1)$. Such a functional will appear the same in every observer's coordinates; as a result, the operator $\hat A$ will appear the same to every observer \emph{without a coordinate transformation}. In the next section, we study a particular case of this type.

\section{Hartle-Hawking functional} \label{sec:HH}

\subsection{The functional in observer-adapted coordinates}

An important example of a de Sitter-invariant functional is the exponent of the Euclidean on-shell action with boundary data $\Pi(\ell)$:
\begin{align}
 \Psi[\Pi(\ell)] \sim e^{-S_E[\Pi(\ell)]} \ , \label{eq:HH_functional}
\end{align}
where we left the normalization unspecified. In ordinary $dS_4$, this functional would be the wavefunction of the Bunch-Davies vacuum \cite{Bunch:1978yq}, or the Hartle-Hawking state \cite{Hartle:1983ai}, in the $\Pi(\ell)$ basis. Indeed, since we are in a free theory, on-shell action fully captures the path integral. In our $dS_4/\bbZ_2$ setup, the functional \eqref{eq:HH_functional} will eventually play the role of a thermal state. For now, let us simply translate it into an operator on some observer's Hilbert space.

The Euclidean action in \eqref{eq:HH_functional} can be calculated as the renormalized action in Euclidean $AdS_4$ (i.e. 4d hyperbolic space) with boundary data $\Pi(\ell)$ \cite{Maldacena:2002vr}. The Euclidean $AdS_4$ can be conveniently represented in the same 4+1d picture as $dS_4$, as the hyperboloid of e.g. future-pointing unit \emph{timelike} vectors. The two spacetimes then explicitly share the same asymptotic 3-sphere of null directions $\ell^\mu$. The appropriate action reads:
\begin{align}
 S_E = \int_{\calM}d^4x \sqrt{g}\left(-\frac{1}{2}(\nabla F)^2 + F^2 \right) + \int_{\del\calM} d^3x\sqrt{q}\left(\frac{1}{2}\dot F^2 + 2F\dot F + F^2 \right) \ . \label{eq:S_off_shell}
\end{align}
Here, $\calM$ is the $EAdS_4$ space truncated at a large but finite 3-sphere $\del\calM$. $\sqrt{g}$ is the $EAdS_4$ 4-volume density, $\sqrt{q}$ is the volume density on $\del\calM$, and $\dot F$ is the derivative of $F$ along an outgoing unit normal to $\del\calM$. The boundary terms in \eqref{eq:S_off_shell} are chosen so that the action's variation is proportional to $\Phi\delta\Pi$. The overall sign is chosen to make the on-shell action positive (see below). On-shell, the action \eqref{eq:S_off_shell} reads:
\begin{align}
 S_E[\Pi(\ell)] = -\frac{1}{2}\int_{\scri} d^3\ell\, \Phi_E(\ell)\Pi(\ell) \ , \label{eq:S_on_shell}
\end{align}
where $\Phi_E(\ell)$ is the Dirichlet boundary data induced by the Neumann data $\Pi(\ell)$ on $EAdS_4$ (recall that in the Lorentzian $dS_4/\bbZ_2$, the Dirichlet data $\Phi(\ell)$ vanishes). We can compute $\Phi_E(\ell)$ using the boundary 2-point function \cite{Witten:1998qj,Freedman:1998tz}: 
\begin{align}
 \Phi_E(\ell) = \int d^3\ell'\, G_E(\ell;\ell')\,\Pi(\ell') \quad ; \quad G_E(\ell;\ell') = \frac{1}{4\pi^2}\cdot\frac{1}{\ell\cdot\ell'} \ . \label{eq:Phi}
\end{align}

The Euclidean action \eqref{eq:S_on_shell} is thus quadratic in the boundary data $\Pi(\ell)$. In the oscillator decomposition \eqref{eq:Pi}, it is constrained by symmetry to take the form:
\begin{align}
 S_E[\Pi(\ell)] = \int_0^\infty\frac{d\omega}{2\pi}\sum_{lm} S_l(\omega)\left|c_{lm}(\omega)\right|^2 \ . \label{eq:S_form}
\end{align}
Using eqs. \eqref{eq:S_on_shell}-\eqref{eq:Phi}, we can compute the coefficients $S_l(\omega)$ as:
\begin{align}
 S_l(\omega)= \left\{\begin{array}{lc}
   \displaystyle \frac{1}{\omega} \tanh\frac{\omega\pi}{2} \prod_{k=1}^{l/2}\frac{(2 k-1)^2+\omega ^2}{(2 k)^2+\omega ^2} \quad & l\text{ even} \\
   \displaystyle \frac{\omega}{1+\omega^2} \coth\frac{\omega\pi}{2}\prod_{k=1}^{(l-1)/2}\frac{(2 k)^2+\omega ^2}{(2 k+1)^2+\omega ^2} \quad & l\text{ odd}
 \end{array}\right. \ . \label{eq:S_result}
\end{align}
The calculation is performed in Appendix \ref{app:Psi}. We conclude that in terms of the normalized oscillators \eqref{eq:a_c},\eqref{eq:N}, the Hartle-Hawking functional \eqref{eq:HH_functional} reads simply:
\begin{align}
 \Psi[\Pi(\ell)] \sim \exp\left(-2\int_0^\infty\frac{d\omega}{2\pi}\sum_{lm} \tanh\left(\frac{\pi}{2}(\omega + il)\right)\left|a_{lm}(\omega)\right|^2 \right) \ , \label{eq:HH_functional_result}
\end{align}
where we made use of the analytical continuation:
\begin{align}
 \tanh\left(x + \frac{l\pi i}{2}\right) = \left\{\begin{array}{cc}\tanh(x) \quad & l\text{ even}\\
    \coth(x) \quad & l \text{ odd}
  \end{array}\right. \ . \label{eq:tan_cot}
\end{align}

\subsection{Transforming into an operator on the observer's Hilbert space}

The Wigner-Weyl transform of Gaussians such as \eqref{eq:HH_functional_result} is well-known, since this is the form of the Wigner distribution for a thermal state:
\begin{align}
 W(a,a^*) \sim \exp\left(-2\tanh\frac{\beta\omega}{2}\left|a\right|^2 \right) \ \longleftrightarrow \ \hat W \sim \exp\left(-\beta\omega\,\hat a^\dagger\hat a \right) \ . \label{eq:thermal}
\end{align}
We use eq. \eqref{eq:thermal}, continued to complex $\beta$, as a purely mathematical statement: so far, the functional \eqref{eq:HH_functional_result} does not have the interpretation of a state. Applying \eqref{eq:thermal} to \eqref{eq:HH_functional_result}, we obtain the operator:
\begin{align}
 \hat\Psi \sim \hat R\, e^{-\pi\hat H} \ , \label{eq:HH_operator}
\end{align}
where $\hat H$ is the observer's Hamiltonian for $t$ translations:
\begin{align}
 \hat H = \int^\infty_0\frac{d\omega}{2\pi}\sum_{lm}\omega\, \hat{a}^\dagger_{lm}(\omega)\, \hat{a}_{lm}(\omega) \ ,
\end{align}
and $\hat R$ is the antipodal operator on the $(\theta,\varphi)$ 2-sphere:
\begin{align}
 \hat R = (-1)^{\hat L} \quad ; \quad \hat L = \int_{0}^{\infty}\frac{d\omega}{2\pi}\sum_{lm} l\,\hat{a}^{\dagger}_{lm}(\omega)\,\hat{a}_{lm}(\omega) \ .
\end{align}
$\hat R$ has eigenvalues $\pm 1$, depending on the number of odd-$l$ quanta. We notice that $\hat\Psi$ is almost, but not quite, the density operator of a thermal state. In particular, due to the factor of $\hat R$, it is not positive definite.

\section{Translating states between observers} \label{sec:states}

\subsection{From a global ``meta-state'' into states for individual observers}

In section \ref{sec:operators}, we've discussed the translation of operators between the different observers' Hilbert spaces. Now, the density matrix of a \emph{state} is a particular case of an operator. Can we, then, translate states between observers by directly applying the operator translation recipe? The answer is no: the translation of operators through the Wigner-Weyl transform will not preserve the non-negativity of eigenvalues, which is a necessary property for a density matrix. 

On the other hand, recall that the translation of operators does preserve Hermiticity, and any Hermitian operator can be made positive-semidefinite by taking the square. Notice also that the operator \eqref{eq:HH_operator}, which we derived from the Hartle-Hawking functional, becomes a thermal density matrix upon squaring. This leads us to propose the following recipe:
\begin{definition} \label{def:meta_to_state}
 Consider a real functional $\Psi: \Gamma\to\bbR$ on the space of boundary data. Such a functional is defined independently from any observer. Given an observer, we translate it into a density matrix $\hat\rho = \hat\Psi^2$ on the observer's Hilbert space, by first applying the Wigner-Weyl transform to get the Hermitian operator $\hat\Psi$, and then squaring this operator to ensure non-negative eigenvalues.
\end{definition}
Definition \ref{def:meta_to_state} postulates a functional $\Psi$ which encodes a ``God's eye'' view of the entire spacetime, while not being a state in any Hilbert space. From this ``meta-state'', we can deduce the (generally, mixed) state that would be visible to each observer. The procedure involves a loss of information, since we lose track of signs when squaring the operator $\hat\Psi$.

For the particular case when the ``meta-state'' is the Hartle-Hawking functional \eqref{eq:HH_functional}, our definition leads, through \eqref{eq:HH_operator}, to the following density matrix for each observer:
\begin{align}
 \hat\rho = \hat\Psi^2 \sim e^{-2\pi\hat H} \ . \label{eq:thermal_dS}
\end{align}
We recognize this as the thermal state at the de Sitter temperature $T = 1/2\pi$.

In the above, we haven't kept track of the normalization of our functionals and states. In particular, the trace of the density matrix \eqref{eq:thermal_dS} can be fixed to $1$ by appropriately normalizing the Hartle-Hawking functional \eqref{eq:HH_functional}. Importantly, this is always possible under Definition \ref{def:meta_to_state}. Indeed, the state's normalization $\tr\hat\rho = \tr(\hat\Psi^2)$ is preserved by the Wigner-Weyl transforms, being the traced product of two $\hat\Psi$'s. By normalizing the ``meta-state'' functional $\Psi$ such that $\Psi^2$ has a unit integral over $\Gamma$, we can ensure that the resulting density matrix for every observer has unit trace. Note that the squaring recipe $\hat\rho = \hat\Psi^2$ is special in this respect: if we attempted to make $\hat\Psi$ positive-semidefinite by e.g. raising it to the \emph{fourth} power, $\tr\hat\rho$ would not be preserved, since it would now be the traced product of more than two $\hat\Psi$'s.

In the semiclassical limit, all operator products reduce to ordinary products of functionals on $\Gamma$. Definition \ref{def:meta_to_state} then produces a classical probability distribution $\rho = \Psi^2$, on which all observers agree. On the other hand, for general quantum states, different observers will disagree on the expectation values of operators (which are in turn translated between the observers' Hilbert spaces through Wigner-Weyl transforms, as in section \ref{sec:operators}). This is because the expectation value $\tr(\hat\rho\hat A) = \tr(\hat\Psi^2\hat A)$ of an operator $\hat A$ is the traced product of \emph{three} operators that undergo Wigner-Weyl transforms. 

The disagreement on expectation values can be already be seen for the observer-independent thermal state \eqref{eq:thermal_dS}. Indeed, this state appears the same to each observer \emph{in his own coordinate system}. The same points in spacetime are assigned different coordinates by the different observers, leading to different correlation functions. Note that the state \eqref{eq:thermal_dS} is \emph{not} semiclassical, since it is near vacuum for high frequencies $\omega$.
\begin{figure}[thp]
\centering
\includegraphics{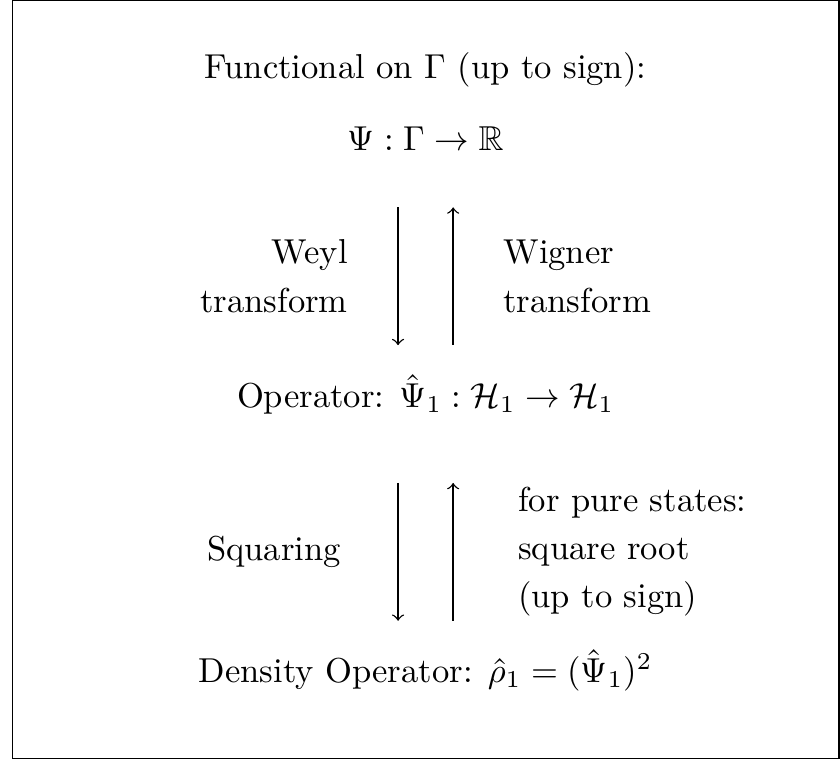}
\caption{Translation recipe for density operators: a functional $\Psi$ gives rise to an operator $\hat{\Psi}_1$ for a specific observer 1. The density operator is then calculated as its square $\hat{\rho}_1=(\hat{\Psi}_1)^2$. If $\hat{\rho}_1$ is pure, we can go backwards to reconstruct $\hat{\Psi}_1$ and $\Psi$ up to an overall sign ambiguity; we can then run the translation forward again, obtaining a mixed state $\hat\rho_2$ in another observer's Hilbert space.}
\label{fig:states}
\end{figure}

\subsection{From a pure state for one observer into a mixed state for another}

The operator squaring in Definition \ref{def:meta_to_state} is in general not invertible. However, it \emph{is} invertible, up to an overall sign, if the density matrix $\hat\rho$ has unit rank, i.e. if the state is pure. This allows us to formulate a translation recipe directly from one observer to another, when the first observer has the full knowledge of a pure state:
\begin{definition} \label{def:state_to_state}
 Consider a pure state, i.e. a rank-1 density matrix $\hat\rho_1$, on the Hilbert space of an observer. We can take the square root $\hat\Psi_1$ of $\hat\rho_1$, up to an overall sign. We then perform the Wigner-Weyl transform into a functional $\Psi$. The latter can be translated into a (generally mixed) state $\hat\rho_2$ for a second observer, following Definition \ref{def:meta_to_state}. The overall sign ambiguity disappears when squaring the operator on the second observer's side.
\end{definition}
As before, the state's normalization is preserved. On the other hand, the purity of the state is not, because the unit rank of the operator $\hat\Psi$ is not preserved by the Wigner-Weyl transforms. In particular, the purity of a state $\hat\rho$ can be measured via $\tr\hat\rho^2 = \tr\hat\Psi^4$; this is not preserved under the Wigner-Weyl transforms, because it is the traced product of more than two $\hat\Psi$'s. Again, we see that the state translation involves information loss. On one hand, if we start with a mixed state, we cannot translate it into another observer's Hilbert space, due to the ambiguous square root. On the other hand, if we start with a pure state, we can perform the translation through Definition \ref{def:state_to_state}, but this results in a mixed state for the other observer. The state translation procedures in Definitions \ref{def:meta_to_state}-\ref{def:state_to_state} are summarized in figure \ref{fig:states}.

\section{Discussion} \label{sec:discuss}

In this paper, we considered a concrete four-dimensional system -- a scalar field in elliptic de Sitter space $dS_4/\bbZ_2$ -- which exhibits the concepts of horizon complementarity. Through this system, we could explicitly study a situation where each observable region is associated with a Hilbert space, but these do not derive from a global Hilbert space for the entire world. We proposed a one-to-one translation recipe between operators on the different Hilbert spaces, based on the Wigner-Weyl transform. We then proposed a translation recipe for states, which comes in two versions. In Definition \ref{def:meta_to_state}, we have a global ``meta-state'' which descends into proper states for individual observers; in Definition \ref{def:state_to_state}, we have a pure state for one observer which gets translated into a mixed state for a second observer, with the ``meta-state'' appearing as an intermediate step. 

Our state-translation recipe preserves the normalization and positive-semidefiniteness of the density matrix, and yields observer-independent predictions in the semiclassical limit. On the other hand, in the quantum regime, it involves information loss, and assigns observer-dependent expectation values to operators. Remarkably, if one starts from the Hartle-Hawking functional \eqref{eq:HH_functional} as the ``meta-state'', one obtains the thermal state \eqref{eq:thermal_dS} at the de Sitter temperature for every observer. Moreover, this is the \emph{only} sensible observer-invariant Gaussian state that one can obtain. Indeed, $S_E[\Pi(\ell)]$ is the only $SO(4,1)$-invariant quadratic functional over the boundary data $\Pi(\ell)$. The remaining freedom in choosing the ``meta-state'' is to take $e^{-\alpha S_E}$ with a constant coefficient $\alpha\neq 1$. However, based on our results for $e^{-S_E}$, it is easy to find that the corresponding state will have diverging temperature $T_l(\omega)\sim \omega$ for the high-frequency modes.

It will be interesting to further investigate the information loss in our state-translation recipe. Can it be related more directly to fact that each observer cannot see his own horizon? We've argued above that the horizon's removal should not be seen as affecting the classical solution space. Similarly, we believe that the information loss at the quantum level should be associated not with the field value at the horizon itself, but with correlations across it: the main effect of removing the horizon is not the absence of a zero-measure set of points, but the change in the spacetime's connectivity. It would be interesting to probe this idea, and to see if it can be related to the notion of horizon entropy as entanglement entropy in the QFT (see e.g. \cite{Bianchi:2012ev}). One route to approach these issues may be to study one observer's vacuum, and work out the resulting mixed states for other observers. 

The general considerations in this paper, as well as the operator/state translation recipes, make sense for any QFT in $dS_4/\bbZ_2$. The exception is the calculation of the thermal state from the Hartle-Hawking functional. Our free massless conformally-coupled scalar \eqref{eq:F_diff} is special in that it has only one of the two types of boundary data -- in this case, the Neumann data $\Pi(\ell)$ -- non-vanishing on $\scri$. Changing the mass or adding interactions ruins this property. Without it, it is unclear which boundary data we should feed into the Hartle-Hawking functional \eqref{eq:HH_functional}, making our calculation ill-defined. On the other hand, free massless gauge fields \emph{do} have the property that only one type of asymptotic data is non-vanishing, for all values of the spin \cite{Neiman:2014npa}. One can therefore repeat our calculation for such fields, as well as for the antipodally-odd version of the scalar \eqref{eq:F_diff}. It would be interesting to know if the $T=1/2\pi$ thermal state is again obtained, and if so, to understand the mechanism behind this.

Unfortunately, our story cannot be generalized to gravitationally perturbed, asymptotically $dS_4/\bbZ_2$ spacetimes. This is because such spacetimes have closed timelike loops \cite{Gao:2000ga,Leblond:2002ns,Parikh:2002py}, which pure $dS_4/\bbZ_2$ avoids marginally. In particular, this prevents us from considering dynamical gravity in asymptotically $dS_4/\bbZ_2$.

There is, however, a tantalizing possibility to extend this paper's insights into a quantum-gravitational setting. That setting is Vasiliev's theory of higher-spin gravity \cite{Vasiliev:1999ba}. This is an interacting theory of massless gauge fields of all spins, along with a massless conformally-coupled scalar. As found recently in \cite{Neiman:2014npa}, this theory satisfies the property that only one type of boundary data for each field in $dS_4/\bbZ_2$ is non-vanishing on $\scri$, order by order in perturbation theory. One can then speak sensibly about the generalization of the (no longer Gaussian) Hartle-Hawking functional \eqref{eq:HH_functional}. Moreover, it is precisely this functional that appears in the holographic dS/CFT proposal \cite{Anninos:2011ui} for higher-spin gravity, as the dual of the CFT partition function. As for the issue of closed timelike loops, it is not obviously a problem in higher-spin gravity, since the interactions are very different from General Relativity, and appear to be non-local. Thus, we have a potential link between a global holographic model on $\scri$ and the thermal states inside observers' causal patches. This would be a great step forward for the application of dS/CFT to the observer-related puzzles in de Sitter space.

\section*{Acknowledgements}

We are grateful to Abhay Ashtekar and Laurent Freidel for discussions. Research at Perimeter Institute is supported by the Government of Canada through Industry Canada and by the Province of Ontario through the Ministry of Research \& Innovation. YN also acknowledges support of funding from NSERC Discovery grants. During the early stages of this work, YN was employed at Penn State University, supported in part by the NSF grant PHY-1205388 and the Eberly Research Funds of Penn State.

\appendix
\section{Normalization of harmonic oscillators on \texorpdfstring{$\scri$}{Scri}} \label{app:N}

In this Appendix, we calculate the normalization coefficients \eqref{eq:N} that relate the harmonic oscillators $(a_{lm}(\omega),a^*_{lm}(\omega))$ with commutation relations \eqref{eq:a_commutators} to the functions $(c_{lm}(\omega),c^*_{lm}(\omega))$ in the Fourier expansion \eqref{eq:Pi} of the boundary data $\Pi(\ell)$. We will do this by calculating the energy of a solution \eqref{eq:Pi} with respect to the observer's time-translation generator $\xi^\mu$, and comparing to eq. \eqref{eq:E}. To calculate the energy, we should integrate the stress tensor over a spatial slice in the static patch. In practice, it is easier to take the limit where the spatial slice becomes the observer's e.g. future horizon.

\subsection{Propagating the field to the horizon}

Let us find the solution $F(x)$ on the future horizon from the asymptotic Neumann data \eqref{eq:Pi}. We use the horizon coordinates $(\eta,\theta,\varphi)$ from section \ref{sec:geometry:coordinates}. From symmetry, we know that $F$ will take the form:
\begin{align}
 F(\eta,\theta,\varphi) 
  = \frac{1}{2}\int^\infty_0\frac{d\omega}{2\pi}\sum_{lm}\left(I_l(\omega)\, e^{-i\omega\eta}\, Y_{lm}(\theta,\varphi)\, c_{lm}(\omega) + c.c.\right) \, , \label{eq:F_I}
\end{align}
for some functions $I_l(\omega)$, where the factor of $1/2$ is chosen for later convenience. To find the coefficients $I_l(\omega)$, we can focus on an $m=0$ mode:
\begin{align}
 \Pi(t,\theta,\varphi) = e^{-i\omega t}\,Y_{l0}(\theta,\varphi)+c.c. = \sqrt{\frac{2l+1}{4\pi}}\,e^{-i\omega t}P_l(\cos\theta)+c.c. \ , \label{eq:single_mode}
\end{align}
where $P_l(u)$ are the Legendre polynomials. Plugging into \eqref{eq:F_I}, we get:
\begin{align}
 F(\eta,\theta,\varphi) = \frac{1}{2}\sqrt{\frac{2l+1}{4\pi}}\,I_l(\omega)\,e^{-i\omega\eta}P_l(\cos\theta) + c.c. \label{eq:F_horizon_l}
\end{align}
On the other hand, we can use the boundary-to-bulk propagator \eqref{eq:F_Pi}-\eqref{eq:propagator} to calculate the field at e.g. the horizon point $x^\mu = (e^\eta,e^\eta,1,0,0)$. For a boundary point parametrized by a null vector $\ell^\mu(t,\theta,\varphi)$ as in \eqref{eq:ell}, the propagator reads:
\begin{align}
 G(x;\ell) = \frac{1}{4\pi}\delta(x\cdot\ell) = \frac{1}{4\pi}\delta(\cos\theta - e^{\eta-t}) \ .
\end{align}
Integrating the propagator over the boundary, we get:
\begin{align}
 F(\eta,0,0) = \frac{1}{2}\sqrt{\frac{2l+1}{4\pi}}e^{-i\omega \eta}\int^1_0 du\, u^{i\omega-1}P_l(u) + c.c. \ . \label{eq:F_eta_00}
\end{align}
Comparing with \eqref{eq:F_horizon_l}, we identify the integral in \eqref{eq:F_eta_00} as our desired coefficient $I_l(\omega)$.

\subsection{Evaluating the integral \texorpdfstring{$I_l(\omega)$}{I(l,omega)}}

Let us now calculate the integral:
\begin{align}
 I_l(\omega) = \int^1_0 u^{i\omega -1}P_l(u)du \ . \label{eq:I_integral}
\end{align}
Using the standard recursion relation for Legendre polynomials:
\begin{align}
 (l+1)P_{l+1}(u) = (2l+1)uP_l(u) - lP_{l-1}(u) \ , \label{eq:Legendre_recursion}
\end{align}
we get the following recursion relation for $I_l(\omega)$:
\begin{align}
 I_{l+1}(\omega) = \frac{2l+1}{l+1}\,I_l(\omega-i) - \frac{l}{l+1}\,I_{l-1}(\omega) \ . \label{eq:I_recursion}
\end{align}
The initial values for this recursion relation can be calculated directly, using the regularization $u^{i\omega -1}\rightarrow u^{i\omega-1+\varepsilon}$ in the integral \eqref{eq:I_integral}:
\begin{align}
 I_0(\omega) = \frac{1}{i\omega} \quad ; \quad I_1(\omega) = \frac{1}{1+i\omega} \ . \label{eq:I_initial}
\end{align}
After calculating $I_l(\omega)$ for the first few $l$'s, we can guess the general expression:
\begin{align}
 I_l(\omega) = \left\{\begin{array}{lc}
    \displaystyle \frac{1}{i\omega} \prod_{k=1}^{l/2}\frac{\omega + (2k-1)i}{\omega - 2ki} \quad & l\text{ even}\\
    \displaystyle \frac{1}{1+i\omega}\prod_{k=1}^{(l-1)/2}\frac{\omega + 2ki}{\omega - (2k+1)i} \quad & l\text{ odd}
  \end{array}\right.\,. \label{eq:I}
\end{align}
This ansatz can be verified by checking that it satisfies the recursion relation \eqref{eq:I_recursion} and agrees with the initial values \eqref{eq:I_initial}.

\subsection{Energy of a given solution}

Having translated the field \eqref{eq:Pi} from $\scri$ to the future horizon $\mathcal{T}^+$, we are able to determine its energy. The stress energy tensor is given by:
\begin{align}
 T_{\mu\nu} = \del_\mu F\del_\nu F + g_{\mu\nu}(\dots) \ . \label{eq:T}
\end{align}
The piece proportional to $g_{\mu\nu}$ will not contribute, because we will contract with a null vector. The energy on the horizon is given by:
\begin{align}
 E = \int_{\mathcal{T}^+} T_{\mu\nu}\,\xi^\mu S^\nu d^3x \ .
\end{align}
Here, the 3d vector density $S^\mu$ is the horizon's area current, and $\xi^\mu$ is the observer's time-translation Killing field, which becomes null on the horizon. Both $S^\mu$ and $\xi^\mu$ point along the horizon's null normal. Specifically, $\xi^\mu$ generates translations in the null coordinate $\eta$. Using the expression \eqref{eq:T} for the stress tensor, this gives:
\begin{align}
 E = \int_{-\infty}^\infty d\eta\int d\Omega\left(\frac{\del F}{\del\eta}\right)^2
   = \frac{1}{2}\int_0^\infty\frac{d\omega}{2\pi}\sum_{lm}\,\omega^2|I_l(\omega)|^2\,|c_{lm}(\omega)|^2\,.
\end{align}
Using the expression \eqref{eq:I} for $I_l(\omega)$ and comparing with eq. \eqref{eq:E}, we obtain the normalization coefficients \eqref{eq:N} for the harmonic oscillators.

\section{Euclidean action in the harmonic oscillator basis} \label{app:Psi}

In this Appendix, we prove the expression \eqref{eq:S_result} for the expansion coefficients of the Euclidean on-shell action \eqref{eq:S_on_shell} in an observer's oscillator basis. 

\subsection{An integral expression}

The coefficients in the action's expansion \eqref{eq:S_form} can be identified with coefficients in the asymptotic Dirichlet data $\Phi_E(\ell)$ that is induced in $EAdS_4$ by our Neumann data \eqref{eq:Pi}. Indeed, from symmetry, $\Phi_E$ must take the form:
\begin{align}
 \Phi_E(t,\theta,\varphi) = -\int^\infty_0\frac{d\omega}{2\pi}\sum_{lm}\left(S_l(\omega)\,e^{-i\omega t}\,Y_{lm}(\theta,\varphi)\,c_{lm}(\omega) + c.c.\right)\,, \label{eq:Phi_B}
\end{align}
for some unknown functions $S_l(\omega)$. We will see below that these functions are real. One can then plug them into the expression \eqref{eq:S_on_shell} for the action, and find that they coincide with the $S_l(\omega)$ in the action's expansion \eqref{eq:S_form}. To find these functions explicitly, we again restrict to the $\Pi(\ell)$ configuration \eqref{eq:single_mode} with a single $m=0$ mode. Plugging into \eqref{eq:Phi_B}, we get:
\begin{align}
 \Phi_E(t,\theta,\varphi) = -\sqrt{\frac{2l+1}{4\pi}}\,e^{-i\omega t}P_l(\cos\theta)S_l(\omega)+c.c. \ .
\end{align}
On the other hand, we can evaluate $\Phi_E$ explicitly using the 2-point function \eqref{eq:Phi}, e.g. at the point $(t,\theta,\varphi) = (t,0,0)$. The 2-point function reads:
\begin{align}
 G_E(\ell;\ell') = \frac{1}{4\pi^2}\cdot\frac{1}{\ell\cdot\ell'} = \frac{1}{4\pi^2(\cos\theta'-\cosh(t-t'))} \ ,
\end{align}
where we used the parametrization \eqref{eq:ell} for the boundary points $\ell(t,0,0)$ and $\ell'(t',\theta',\varphi')$. Integrating this over $\ell'$, we get:
\begin{align}
 \Phi_E(t,0,0) = \frac{1}{4\pi^2}\sqrt{\frac{2l+1}{4\pi}}\int^\infty_{-\infty}dt'\int d\Omega'\,\frac{e^{-i\omega t'}P_l(\cos\theta)}{\cos \theta'-\cosh(t-t')}+c.c. \ .
\end{align}
The integration over $\varphi'$ gives a factor of $2\pi$. We then perform the substitutions $u=\cos\theta'$ and $\tau=t'-t$, which give:
\begin{align}
 \Phi_E(t,0,0) = \frac{1}{2\pi}\sqrt{\frac{2l+1}{4\pi}}\,e^{-i\omega t}\int^\infty_{-\infty}d\tau\int^1_{-1} du\,\frac{e^{-i\omega\tau}P_l(u)}{u-\cosh\tau}+c.c. \ , 
\end{align}
from which we read off:
\begin{align}
 S_l(\omega) = \frac{1}{2\pi}\int^\infty_{-\infty}d\tau\int^1_{-1} du\,\frac{e^{-i\omega\tau}P_l(u)}{\cosh\tau - u} \ , \label{eq:S_l_integral}
\end{align}
which is indeed real after the $\tau$ integration.

\subsection{Evaluating the integral}

Let us now prove that the integral \eqref{eq:S_l_integral} evaluates to \eqref{eq:S_result}. First, for any function $f(\omega)$, we define the Fourier transform and its inverse as:
\begin{align}
 \widetilde f(\tau):=\int_{-\infty}^{\infty}f(\omega)e^{i\omega\tau}d\omega \quad ; \quad 
 f(\omega)=\frac{1}{2\pi}\int_{-\infty}^{\infty}\widetilde f(\tau)e^{-i\omega\tau}d\tau \ .
\end{align}
Thus, the Fourier transform of \eqref{eq:S_l_integral} reads:
\begin{align}
 \widetilde S_l(\tau) = \int^1_{-1} \frac{P_l(u)\,du}{\cosh\tau - u} \ . \label{eq:Fourier_integral}
\end{align}
It remains to prove that this coincides with the Fourier transform of \eqref{eq:S_result}. We will do this by recursion in $l$. For $l=0,1$, we can evaluate \eqref{eq:Fourier_integral} explicitly, as:
\begin{align}
 \widetilde S_0(\tau) = 2\log\left(\coth\frac{\tau}{2}\right) \quad ; \quad
 \widetilde S_1(\tau) = 2\cosh\tau\log\left(\coth\frac{\tau}{2}\right)-2 \ . \label{eq:S_initial}
\end{align}
We can then use the recursion relation \eqref{eq:Legendre_recursion} for the Legendre polynomials, along with the identity (for $l>0$):
\begin{align}
\int^1_{-1}\frac{uP_l(u)\,du}{\cosh\tau-u}=\cosh\tau\,\int^1_{-1}\frac{P_l(u)\,du}{\cosh\tau-u} \ ,
\end{align}
to get the following recursion relation for the integrals \eqref{eq:Fourier_integral}:
\begin{align}
 (l+1)\widetilde S_{l+1}(\tau)=(2l+1)\cosh\tau\,\widetilde S_l(\tau)-l\,\widetilde S_{l-1}(\tau)\,. \label{eq:S_recursion}
\end{align}

It remains to show that the Fourier transform of \eqref{eq:S_result} also satisfies the recursion relation \eqref{eq:S_recursion} and the initial conditions \eqref{eq:S_initial}. Let us denote the expression \eqref{eq:S_result} as $\underline{S_l}(\omega)$, to distinguish it from the expression \eqref{eq:S_l_integral} with which we are comparing. We can calculate the Fourier transform of $\underline{S_l}(\omega)$ explicitly via Cauchy's theorem. 

First, we notice that $\underline{S_l}(\omega)$ is symmetric in $\omega$, ensuring that $\underline{\widetilde S_l}(\tau)$ is symmetric in $\tau$, as it should be to agree with \eqref{eq:Fourier_integral}. We can therefore restrict to $\tau>0$. To compute the Fourier transform, we notice that $\underline{S_l}(\omega)$ is a meromorphic function, with behavior $\underline{S_l}(\omega)\sim \omega^{-1}$ for $|\omega|\to\infty$ in the complex plane. Therefore, we have:
\begin{align}
 \underline{\widetilde S_l}(\tau) = \int^\infty_{-\infty}\underline{S_l}(\omega)\,e^{i\omega\tau}d\omega=\int_{\mathcal{C}}\underline{S_l}(\omega)\,e^{i\omega\tau}d\omega \ ,
 \label{eq:Fourier_transform_contour}
\end{align}
where $\mathcal{C}$ is a closed half-circle around the upper half of the complex plane. $\underline{S_l}(\omega)$ has only simple poles, located at $\omega=(2n+1)i$ for even $l$ and at $\omega=2ni$ for odd $l$, where $n\in\bbZ$. Summing up the residues in the upper half-plane, we find that the Fourier transform \eqref{eq:Fourier_transform_contour} is given by:
\begin{align}
 \underline{\widetilde S_l}(\tau) = \sum^\infty_{n=0}\left\{\begin{array}{lc}
  \displaystyle \frac{4e^{-(2n+1)\tau}}{2n+1} \prod_{k=1}^{l/2}\frac{(2 k-1)^2-(2n+1)^2}{(2 k)^2-(2n+1)^2} \quad & l \text{ even}\\
  \displaystyle \frac{8ne^{-2n\tau}}{4n^2-1} \prod_{k=1}^{(l-1)/2}\frac{(2 k)^2-(2n)^2}{(2 k+1)^2-(2n)^2} \quad & l \text{ odd}
 \end{array}\right. \ . \label{eq:Fourier_transform_result}
\end{align}\par
We can now explicitly check that this satisfies the recursion relation \eqref{eq:S_recursion} with initial conditions \eqref{eq:S_initial}. Checking the recursion relation is a straightforward exercise. For the initial conditions, we find:
\begin{align}
 \begin{split}
   \underline{\widetilde S_0}(\tau) &= \sum^\infty_{n=0}\frac{4e^{-(2n+1)\tau}}{2n+1} = 4\artanh(e^{-\tau}) \ ; \\
   \underline{\widetilde S_1}(\tau)&=\sum^\infty_{n=0}\frac{8ne^{-2n\tau}}{4n^2-1} = 4\cosh\tau\artanh(e^{-\tau}) - 2 \ .
\end{split}
\end{align}
For $\tau>0$, we have $2\artanh(e^{-\tau}) = \log\left(\coth\frac{\tau}{2}\right)$, which proves that $\underline{\widetilde S_l}(\tau)$ indeed satisfies the initial conditions \eqref{eq:S_initial}.

\end{document}